\documentclass[11pt]{article}
\usepackage{amsmath,amsfonts,amssymb,amsthm}
\usepackage{fullpage}
\usepackage{xspace}
\usepackage{epstopdf}
\usepackage{caption}
\usepackage{tikz}
\usepackage[boxruled,linesnumbered,noline]{algorithm2e}

\newtheorem{theorem}{Theorem}
\newtheorem{lemma}{Lemma}

\newtheorem{problem}{Problem}

\renewcommand{\thefootnote}{\fnsymbol{footnote}}

\def\calT{\mathcal{T}}

\def\calO{\mathcal{O}}

\def\mcenter{{\sf MatCenter}}
\def\rmcenter{{\sf Robust-MatCenter}}
\def \multiknap{{\sf KnapCenter}\xspace}
\def \rknap{{\sf Robust-KnapCenter}\xspace}
\def \groupknap{group multi-knapsack\xspace}

\newcommand{\commentout}[1]{}
\newcommand{\eat}[1]{}

\newcommand{\calH}{{\mathcal H}}

\newcommand{\calS}{{\mathcal S}}

\newcommand{\calI}{{\mathcal I}}
\newcommand{\calM}{{\mathcal M}}

\newcommand{\dist}{d}

\newcommand{\opt}{\mathsf{OPT}}

\newcommand{\B}{\mathsf{B}}
\newcommand{\E}{\mathsf{E}}
\newcommand{\I}{S'}
\renewcommand{\O}{\mathsf{O}}

\renewcommand{\P}{{\rm P}\xspace}
\newcommand{\NP}{{\rm NP}\xspace}
\newcommand{\kb}{{\mathcal B}}
\newcommand{\V}{V}

\newcommand{\true}{$\mathrm{TRUE}$}

\begin{document}

\title{Matroid and Knapsack Center Problems\footnotemark[1]}
\author{Danny Z.~Chen\footnotemark[2] \and Jian Li\footnotemark[3] \and Hongyu Liang\footnotemark[4]
\and Haitao Wang\footnotemark[5]}

\renewcommand{\thefootnote}{\fnsymbol{footnote}}

\footnotetext[1]{This work was supported in part by the National
Basic Research Program of China Grant 2011CBA00300, 2011CBA00301,
and the National Natural Science Foundation of China Grant 61033001,
61061130540, 61073174, 61202009. The research of D.Z. Chen was supported in part by NSF under
Grants CCF-0916606 and CCF-1217906.}

\footnotetext[2]{Department of Computer Science and Engineering, University of Notre Dame, Notre Dame, IN 46556, USA.
  E-mail: dchen@cse.nd.edu}

\footnotetext[3]{Institute for Interdisciplinary Information
Sciences, Tsinghua University, Beijing, 100084, China. E-mail:
lijian83@mail.tsinghua.edu.cn}

\footnotetext[4]{Institute for Interdisciplinary Information
Sciences, Tsinghua University, Beijing, 100084, China. E-mail:
lianghy08@mails.tsinghua.edu.cn}

\footnotetext[5]{Department of Computer Science, Utah State University, Logan, UT 84322, USA.
E-mail: haitao.wang@usu.edu}

\date{}
\maketitle


\begin{abstract}
In the classic $k$-center problem, we are given a metric graph,
and the objective is to
select $k$ nodes as centers such that the maximum distance from any vertex
to its closest center is minimized.
In this paper, we consider two important
generalizations of $k$-center, the matroid center problem
and the knapsack center problem.
Both problems are motivated by recent
content distribution network applications.
Our contributions can be summarized as follows:
\begin{enumerate}
\item We consider the matroid center problem in which
the centers are required to form an independent set of a given matroid.
We show this problem is NP-hard even on a line. We present
a 3-approximation algorithm for the problem on general metrics.
We also consider the outlier version of the problem where
a given number of vertices can be excluded as outliers from the solution.
We present a 7-approximation for the outlier version.
\item We consider the (multi-)knapsack center problem in which
the centers are required to satisfy one (or more) knapsack constraint(s).
It is known that the knapsack center problem with a single knapsack constraint
admits a 3-approximation.
However, when there are at least two knapsack constraints,
we show this problem is not approximable at all.
To complement the hardness result, we present a polynomial time algorithm that gives
a 3-approximate solution such that one knapsack constraint is satisfied and the others
may be violated by at most a factor of $1+\epsilon$.
We also obtain a 3-approximation for the outlier version
that may violate the knapsack constraint by $1+\epsilon$.
\end{enumerate}
\end{abstract}

\newpage

\section{Introduction}
\label{sec:intro}

The $k$-center problem is a fundamental
facility location problem. In the basic version, we are given a metric space $(V,d)$
and are asked to locate a set $\calS \subseteq V$ of at most $k$ vertices as centers and
to assign the other vertices to the centers,
so as to minimize the maximum distance from any vertex to its assigned center, or more formally, to minimize
$\max_{v\in V} \min_{u\in \calS} d(v,u)$.
In the {\em demand} version of the $k$-center problem,
each vertex $v$ has a positive demand $r(v)$, and
our goal is to minimize the maximum weighted distance from any vertex to the centers,
i.e., $\max_{v\in V} \min_{u\in S} r(v)d(v,u)$.
It is well known that the $k$-center problem is NP-hard and admits a polynomial time 2-approximation
even for the demand version \cite{gonzalez1985clustering,hochbaum1985best}, and that no polynomial time $(2-\epsilon)$-approximation algorithm exists unless $\P = \NP$ \cite{gonzalez1985clustering}.

In this paper, we conduct a systematic study on two generalizations of the $k$-center problem and their variants.
The first one is the {\em matroid center} problem, denoted by \mcenter,
which is almost the same as the $k$-center problem except that, instead of the cardinality constraint on the set of centers, now the centers are required to form an independent set of
a given matroid.
A finite matroid $\calM$ is a pair $(V,\calI)$, where $V$ is a finite set
(called the {\it ground set}) and $\calI$ is
a collection of subsets of $V$.
Each element in $\calI$ is called an {\em independent set}.
Moreover, ${\cal M}=(V,\calI)$ satisfies the following three properties: (1)
$\emptyset\in \calI$; (2) if $A\subseteq B$ and $B\in \calI$, then $A\in \calI$;
(3) for all $A, B \in \calI$ with $|A|>|B|$,
there exists an element $e\in A \setminus B$
such that $B\cup \{e\} \in \calI$.
Following the conventions in the literature, we assume the matroid $\calM$ is given by an independence oracle which,
given a subset $S\subseteq V$, decides whether $S\in \calI$.
For more information about the theory of matroids, see, e.g., \cite{schrijver03}.

The second problem we study is the {\em knapsack center} problem (denoted as \multiknap),
another generalization of $k$-center in which the chosen centers are subject to (one or more) knapsack constraints.
More formally, in \multiknap,  there are $m$ nonnegative weight functions
$w_1,\ldots, w_m$ on $V$, and $m$ weight budgets
$\kb_1,\ldots,\kb_m$. Let $w_i(\V'):=\sum_{v\in \V'}w_i(v)$
for all $\V' \subseteq \V$. A solution takes a
set of vertices $\calS \subseteq \V$ as centers such that
$w_i(\calS) \leq \kb_i$ for all $1\leq i\leq m$.
The objective is still to
minimize the maximum service cost of any vertex in $V$ (the service cost of $v$ equals $\min_{c\in \calS}d(v,c)$, or $\min_{c\in\calS}r(v)d(v,c)$ in the demand version).
In this paper, we are interested only in the case where the number $m$ of knapsack constraints is a constant.
We note that the special case with only one knapsack constraint was studied in \cite{jacm86} under the name of weighted $k$-center, which already generalizes the basic $k$-center problem.

Both \mcenter\ and \multiknap\ are motivated by important applications in
content distribution networks \cite{hajiaghayi2011budgeted,krishnaswamy2011matroid}.
In a content distribution network,
there are several types of servers and a set of clients to be connected to the servers.
Often there is a budget constraint on the number of deployed servers of each type \cite{hajiaghayi2011budgeted}.
We would like to deploy a set of servers subject to these budget constraints in order to minimize the maximum
service cost of any client. The budget constraints correspond to finding an independent set in
a partition matroid.\footnote{
Let $B_1,B_2,\ldots, B_b$ be a collection of disjoint subsets of $V$ and $d_i$ be integers such that $1\leq d_i\leq |B_i|$ for all
$1\leq i\leq b$. We say a set $I\subseteq V$ is independent if $|I\cap B_i|\leq d_i$ for $1\leq i\leq b$.
All such independent sets form a partition matroid.
}
We can also use a set of knapsack constraints to capture the budget constraints for all types
(we need one knapsack constraint for each type).
Motivated by such applications,
Hajiaghayi et al.~\cite{hajiaghayi2011budgeted} first studied the red-blue
median problem in which there are
 two types (red and blue) of facilities, and the goal is to deploy at most
 $k_r$ red facilities and $k_b$ blue facilities so as to minimize the sum of service costs.
Subsequently, Krishnaswamy et al. \cite{krishnaswamy2011matroid}
introduced a more general {\em matroid median} problem which seeks to select
a set of facilities that is an independent set in a given matroid
and the {\em knapsack median} problem
in which the set of facilities must satisfy a knapsack constraint.
The work mentioned above uses the sum of service costs as the objective (the $k$-median objective),
while our work aims to minimize the maximum services cost (the $k$-center objective), which is
another popular objective in the clustering and network design literature.

\subsection{Our Results}
For \mcenter,
we show the problem is \NP-hard to approximate within a factor of
$2-\epsilon$ for any constant $\epsilon>0$, even on a line.
Note that the $k$-center problem on a line can be solved exactly
in polynomial time \cite{ref:ChenEfkCenter11}.
We present a 3-approximation algorithm for \mcenter\ on general
metrics.  This improves the constant factors implied
by the approximation algorithms for matroid
median \cite{krishnaswamy2011matroid,charikar2011dependent}
(see Section~\ref{subsec:mcenterapprox} for details).

Next, we consider the outlier version of \mcenter, denoted as \rmcenter, where
one can exclude at most $n-p$ nodes as outliers.
We obtain a 7-approximation for \rmcenter.
Our algorithm is a nontrivial generalization of
the greedy algorithm of Charikar et al. \cite{charikar2001algorithms},
which only works for the outlier version of the basic $k$-center.
However, their algorithm and analysis do not extend to our problem.
In their analysis, if at least $p$ nodes are covered by $k$ disks (with radius 3 times $\opt$),
they have found a set of $k$ centers and obtained a 3-approximation.
However, in our case, we may not be able to open enough centers
in the covered region, due to the matroid constraint. Therefore, we need to search for centers
globally. To this end, we carefully construct two matroids and argue that their intersection
provides a desirable answer
(the construction is similar to that for the non-outlier version, but more involved).

We next deal with the \multiknap problem. We show that for any
$f>0$, the existence of an $f$-approximation algorithm for
\multiknap with more than one knapsack constraint implies $\P =
\NP$. This is a sharp contrast with the case with only one knapsack
constraint, for which a 3-approximation exists \cite{jacm86} and is
known to be optimal \cite{journals/jacm/ChuzhoyGHKKKN05}. Given this
strong inapproximability result, it is then natural to ask whether
efficient approximation algorithms exist if we are allowed to
slightly violate the constraints. We answer this question
affirmatively. We provide a polynomial time algorithm that, given an
instance of \multiknap with a constant number of knapsack
constraints, finds a 3-approximate solution that is guaranteed to
satisfy one constraint and violate each of the others by at most a
factor of $1+\epsilon$ for any fixed $\epsilon>0$. This generalizes
the result of \cite{jacm86} to the multi-constraint case. Our
algorithm
also works for the demand version of the problem.

We then consider the outlier version of the knapsack center problem, which we denote by \rknap.
We present a 3-approximation algorithm for \rknap that violates the knapsack constraint by a factor of
$1+\epsilon$ for any fixed $\epsilon>0$. Our algorithm can be regarded as a ``weighted'' version of the greedy algorithm of Charikar et al. \cite{charikar2001algorithms} which only works for the unit-weight case.
However, their charging argument does not apply to the weighted case. We instead adopt a more involved algebraic approach to prove the performance guarantee.
We translate our algorithm into inequalities involving point sets, and then directly manipulate the inequalities to establish our desired approximation ratio. The total weight of our chosen centers may exceed the budget by the maximum weight of any client, which can be turned into a $1+\epsilon$ multiplicative factor by the partial enumeration technique.
We leave open the question whether there is a constant factor approximation for \rknap that satisfies the knapsack constraint.

\subsection{Related Work}

For the basic $k$-center problem, Hochbaum and Shmoys
\cite{hochbaum1985best,jacm86} and Gonzalez
\cite{gonzalez1985clustering} developed 2-approximation algorithms,
which are the best possible if P $\ne$ NP
\cite{gonzalez1985clustering}. The former algorithms are based on
the idea of the threshold method, which originates from \cite{journals/jct/EF70}. On some special
metrics like the shortest path metrics on trees, $k$-center (with or
without demands) can typically be solved in polynomial time by
dynamic programming. By exploring additional structures of the
metrics, even linear or quasi-linear time algorithms can be
obtained; see e.g.
\cite{ref:ChenEfkCenter11,ref:ColeSl87,ref:FredericksonPa91} and the
references therein. Several generalizations and variations of
$k$-center have also been studied in a variety of application
contexts; see, e.g. \cite{aggarwal2006achieving,li2010clustering,khuller1997fault,conf/sirocco/ChechikP12,conf/focs/CyganHK12,conf/approx/KhullerSS12}.

A problem closely related to $k$-center is
the well-known $k$-median problem, whose
objective is to minimize the sum of service costs of all nodes instead of the maximum one.
Hajiaghayi et al.~\cite{hajiaghayi2011budgeted} introduced the red-blue
median problem that generalizes $k$-median, and presented a constant factor approximation based on local search.
Krishnaswamy et al. \cite{krishnaswamy2011matroid}
introduced the more general matroid median problem and presented a $16$-approximation algorithm based on LP rounding,
whose ratio was improved to $9$ by Charikar and Li \cite{charikar2011dependent}
using a more careful rounding scheme.
Another generalization of $k$-median is the knapsack median problem studied by Kumar \cite{conf/soda/Kumar12}, which requires to open a set of centers with a total weight no larger
than a specified value. Kumar gave a (large) constant factor approximation for knapsack median, which was improved by Charikar and Li \cite{charikar2011dependent} to a 34-approximation. Several other classical problems have also been investigated recently under matroid or knapsack constraints, such as minimum spanning tree \cite{conf/soda/Zenklusen12}, maximum matching \cite{conf/esa/GrandoniZ10}, and submodular maximization \cite{conf/stoc/LeeMNS09,conf/stoc/VondrakCZ11}.

For the $k$-center formulation, it is well known that a few distant vertices (outliers) can disproportionately affect the final solution.
Such outliers may significantly increase the cost of the solution, without improving the level of service to the majority of clients.
To deal with outliers, Charikar et al. \cite{charikar2001algorithms} initiated the study of the robust versions of $k$-center and other related problems, in which a certain number of points can be excluded as outliers.
They gave a 3-approximation for robust $k$-center, and showed that the problem with forbidden centers
(i.e., some points cannot be centers)
is inapproximable within $3-\epsilon$ unless P = NP. For robust $k$-median, they presented a bicriteria approximation algorithm that returns a $4(1+1/\epsilon)$-approximate solution in which the number of excluded outliers may violate the upper bound by a factor of $1+\epsilon$. Later, Chen \cite{conf/soda/Chen08} gave a truly constant factor approximation (with a very large constant) for the robust $k$-median problem.
McCutchen and Khuller \cite{matthew2008streaming} and Zarrabi-Zadeh and
Mukhopadhyay \cite{zarrabi2009streaming} considered the robust $k$-center problem in a streaming context.

\section{The Matroid Center Problem}\label{sec:mcenter}
In this section, we consider the matroid center problem and its outlier version.
A useful ingredient of our algorithms is the {\em (weighted) matroid intersection} problem defined as follows.
We are given two matroids $\calM_1(V,\calI_1)$ and $\calM_2(V, \calI_2)$ defined on the same ground set $V$.
Each element $v\in V$ has a weight $w(v)\geq 0$.
The goal is to find a common independent set $S$ in the
two matroids, i.e., $S\in \calI_1\cap \calI_2$, such that
the total weight $w(S)=\sum_{v\in S}w(v)$ is maximized.
It is well known that this problem can be solved in polynomial time (e.g., see \cite{schrijver03}).

\subsection{NP-hardness of Matroid Centers on a Line}
\label{subsec:mcenternp}

In contrast to the basic $k$-center problem on a line which can be
solved in near-linear time \cite{ref:ChenEfkCenter11}, we show that \mcenter\ is NP-hard even on a line.
We actually prove the following stronger theorem.

\begin{theorem}
\label{thm:mcenternp}
It is \NP-hard to approximate \mcenter\ on a line within a factor
strictly better than 2, even when the given matroid is a partition matroid.
\end{theorem}

\begin{proof}
In a partition matroid, each element in the ground set is colored using one of the $h$ colors and we are given $h$
integers $b_1,b_2,\ldots, b_h$.
The collection of all independent sets is defined to be all subsets that
contain at most $b_1$ elements of color $1$, at most $b_2$ elements of
color $2$, and so on.

\begin{figure}[h]
\begin{center}
\includegraphics[width=0.8\linewidth]{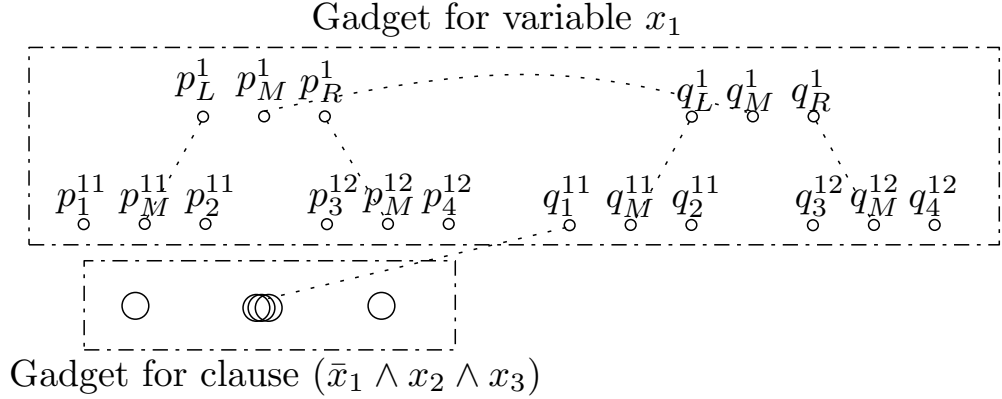}
\caption{\footnotesize A variable gadget and a clause gadget.}\label{fig:npc}
\end{center}
\end{figure}

We use the 3SAT problem for the reduction.
Without loss of generality, we assume that each literal
(including all variables $x_i$ and their negation $\bar{x}_i$)
appears exactly four times in the 3DNF.
Given a 3DNF, we create a \mcenter\ instance as follows.
The points appear in groups. Each group consists of $r$ ($r\geq 3$) points with $r-2$ points in the middle, one to the left and
one to the right. The left and right points are $1$ unit distance away from the midpoints.
Different groups are very far away from each other.
Therefore, in order to make the maximum radius at most one, we need to either select one of the midpoints
in each group or select at least the two points not in the middle.
For each variable $x_i$, we create a variable gadget as follows.
The gadget consists of 6 groups, each having 3 points:
$$
(p^i_L, p^i_M, p^i_R), (q^i_L, q^i_M, q^i_R),
(p^{i1}_{1}, p^{i1}_M, p^{i1}_{2}), (p^{i2}_{3}, p^{i2}_M, p^{i2}_{4}),
(q^{i1}_{1}, q^{i1}_M, q^{i1}_{2}), (q^{i2}_{3}, q^{i2}_M, q^{i2}_{4}).
$$
For two points $p$ and $q$,
we use $[p,q]$ to indicate that we assign a new color to $p$ and $q$.
The color assignment for the gadget is defined by the following pairs:
$$
[p^i_M, q^i_M], [p^i_L, p^{i1}_M], [p^i_R, p^{i2}_M],
 [q^i_L, q^{i1}_M], [q^i_R, q^{i2}_M].
 $$
We are allowed to choose at most one point as a center from each color class.
Points $p^{i1}_{1}, p^{i1}_{2}, p^{i2}_{3}, p^{i2}_{4}$ are called {\em positive portals} of $x_i$ and
points $q^{i1}_{1}, q^{i1}_{2}, q^{i2}_{3}, q^{i2}_{4}$ are called {\em negative portals} of $x_i$.
See Figure~\ref{fig:npc} for an example.
For each clause, we create a clause gadget, which is a group of $5$ points.
We have 3 points in the middle (co-located at the same place), each corresponding to a literal in the clause.
If the point corresponds to a positive  (negative) literal, say $x_i$ (or $\bar{x}_i$),
the point is paired with one of the positive (negative) portals of $x_i$ and we assign
the pair a new color.
We also require that at most one point can be chosen as a center in this pair.
Each portal can be paired at most once.
Since each literal appears exactly 4 times, we have enough portals for the clause gadgets.
All the left and right points of all clause gadgets have the same color
but we are allowed to choose none of them as centers.

We can show that the optimal radius for the \mcenter\ instance is $1$ if and only if the 3DNF formula is
satisfiable. First, suppose the 3DNF is satisfiable.
If $x_i$ is \true\ in a truth assignment, then
we pick $p^i_M, p^{i1}_M, p^{i2}_M$ and $p^{i1}_{1}, p^{i1}_{2}, p^{i2}_{3}, p^{i2}_{4}$ as centers.
Otherwise, we pick
$q^i_M, q^{i1}_M, q^{i2}_M$ and $q^{i1}_{1}, q^{i1}_{2}, q^{i2}_{3}, q^{i2}_{4}$ as centers.
It is straightforward to verify the independence property.
For each group, at least one of the midpoints is selected. Thus, the optimal solution is $1$.
Given the correspondence, the reverse direction can be proved similarly and we omit it.
\end{proof}

\subsection{A $3$-Approximation for \mcenter}
\label{subsec:mcenterapprox}

In fact, we can obtain
a constant approximation for \mcenter\ by using the
constant approximation for the matroid median problem
\cite{krishnaswamy2011matroid,charikar2011dependent}, which roughly gives a 9-approximation for \mcenter.
The idea is given below.

We
say a space $V$ with a distance function $d$ satisfies the
\emph{$(\lambda,c)$-relaxed triangle inequality} (TI) for some
$\lambda$ and $c$, if $d(a_0,a_c)\leq
\lambda\sum_{i=1}^{c}d(a_{i-1},a_i)$ for all $a_0,a_1,\ldots,a_c \in
V$. (Thus a metric space satisfies the $(1,c)$-relaxed TI for all
$c\geq 1$.) By examining the algorithms in \cite{krishnaswamy2011matroid,charikar2011dependent} for the matroid median problem, we notice that they can actually
give a $(\mu\lambda)$-approximation for matroid median where $\mu$ is some universal constant, if the underlying space satisfies the $(\lambda,c_0)$-relaxed
TI for some algorithm-dependent $c_0$.\footnote{We note that Golovin
et al. \cite{conf/fsttcs/GolovinGKT08} claimed (without a proof) that, in
our notations, most existing approximation algorithms for $k$-median
achieve an $O(\lambda)$-approximation on spaces satisfying
$(\lambda,2)$-relaxed TI.
By a scrutiny of the existing $k$-median algorithms, we are not able to reproduce
the same result and the correct approximation ratio should be roughly $O(\lambda^{c_0})$.
However, the results of \cite{conf/fsttcs/GolovinGKT08} are not affected in any essential way
since this only changes the constant hidden in the big-oh notation.
}
(Roughly speaking, $c_0$ is the maximum number of times that the
triangle inequality is used for bounding the distance between a client
and a facility.)
Now, given an instance of \mcenter\ with metric
space $(V, d)$, we define a new distance function $d'$ as
$d'(a,b)=(d(a,b))^p$ for all $a,b\in V$, where $p>2$ is a parameter
whose value will be specified later. By the convexity of the
function $f(x)=x^{p}$ when $p\geq 2$, for all $c\geq 1$ and
$a_0,a_1,\ldots,a_c \in V$, we have
$(\sum_{i=1}^{c}d(a_{i-1},a_i)/c)^p \leq
\sum_{i=1}^{c}d(a_{i-1},a_i)^p/c$, and thus
\begin{eqnarray*}
d'(a_0,a_c)
&=&d(a_0,a_c)^p \leq(\sum_{i=1}^{c}d(a_{i-1},a_i))^{p} \\
&\leq&
c^{p-1}\sum_{i=1}^{c}d(a_{i-1},a_i)^p
=c^{p-1}\sum_{i=1}^{c}d'(a_{i-1},a_i).
\end{eqnarray*}
Therefore $(V, d')$ satisfies the $(c^{p-1},c)$-relaxed TI for all
$c \geq 1$. In particular, it satisfies the
$(c_0^{p-1},c_0)$-relaxed TI where $c_0$ is the algorithm-dependent
parameter mentioned before. We now solve the matroid median problem
on the instance with the new distance function $d'$. Let $\opt$
denote the optimal objective value of \mcenter\ on the original
instance. Then it is clear that the optimal cost of matroid median
on the new instance is at most $|V|\cdot \opt^{p}$. By our previous
observation, the algorithms of
\cite{krishnaswamy2011matroid,charikar2011dependent} give a solution
of cost at most $\mu c_0^{p-1}|V|\opt^{p}$. Transforming the
distance function back to $d$, the maximum service cost of any
client is at most $(\mu
c_0^{p-1}|V|\opt^{p})^{1/p}=c_0^{1-1/p}(\mu|V|)^{1/p}\opt$. By
choosing $p=\Omega(|V|)$, this can produce a
$(c_0+\epsilon)$-approximation for \mcenter\ for any fixed $\epsilon>0$. Using the algorithm of
\cite{charikar2011dependent} this roughly gives a 9-approximation.

\begin{algorithm}[t]
  \caption{Algorithm for \mcenter\ on $G_i$} \label{alg:mcenter}
  Initially, $C \leftarrow \emptyset$, and mark all vertices in $V$ as uncovered.

  \While{$V$ contains uncovered vertices}
  {
  Pick an uncovered vertex $v$. Set $\B(v)\leftarrow \B(v,\dist(e_i))$ and $C\leftarrow C\cup \{v\}.$

  Mark all vertices in $\B(v, 2\dist(e_i))$ as covered.
  }

Define a partition matroid $\calM_\B=(V,\calI)$ with
partition $\{\{\B(v)\}_{v\in C}, V\setminus
\cup_{v\in C}\B(v)\}$ (note that $\{\B(v)\}_{v\in C}$ are disjoint sets by Lemma \ref{lm:disjoint}), where $\calI$ is the set of subsets of $V$ that contains at most 1 element from every $\B(v)$ and 0
element from $V\setminus \cup_{v\in C}\B(v)$.

Solve the unweighted (or, unit-weight) matroid intersection problem between $\calM_\B$
and $\calM$ to get an optimal intersection $\calS$. If
$|\calS|<|C|$, then we declare a failure and try the next $G_i$.
Otherwise, we succeed and return $\calS$ as the set of centers.
\end{algorithm}

We next present a $3$-approximation for \mcenter, thus
improving the ratio derived from the matroid median algorithms \cite{krishnaswamy2011matroid,charikar2011dependent}. Also,
compared to their LP-based algorithms, ours is
simpler, purely combinatorial, and easy to implement.
We begin with the description of our algorithm. Regard the metric space as a (complete) graph $G=(V,E)$ where each edge $\{u,v\}$ has length $d(u,v)$.
Let $\B(v, r)$ be the set of vertices that are at most $r$ unit distance away from $v$ (it depends on the underlying graph).
Let $e_1, e_2, \dots, e_{|E|}$ be
the edges in a non-decreasing order of their lengths. We
consider each
spanning subgraph $G_i$ of $G$ that contains only the first $i$ edges, i.e., $G_i=(V, E_i)$ where $E_i=\{e_1,\dots,e_i\}$.
We run
Algorithm \ref{alg:mcenter} on each $G_i$ and take the best solution.

\begin{lemma}
\label{lm:disjoint}
For any two distinct $u,v\in C$, $\B(u)$ and $\B(v)$ are disjoint sets.
\end{lemma}
\begin{proof}
Suppose we are working on $G_i$ and
there is a node $w$ that is in both $\B(u)$ and $\B(v)$.
Then we know $\dist(w,u)\leq \dist(e_i)$ and $\dist(w,v)\leq \dist(e_i)$.
Thus, $\dist(u,v)\leq 2\dist(e_i)$. But this contradicts with
the fact that the distance between
every two nodes in $C$ must be larger than $2\dist(e_i)$.
\end{proof}

\begin{theorem}
Algorithm \ref{alg:mcenter} produces a $3$-approximation for \mcenter.
\end{theorem}

\begin{proof}
Suppose the maximum radius of any cluster in an optimal solution is $r^*$
and a set of optimal centers is $C^*$. Consider the algorithm on $G_i$ with $\dist(e_i) = r^*$
($r^*$ must be the length of some edge).
First we claim that there exists an intersection of $\calM$ and $\calM_\B$
of size $|C|$.
In fact, we show there is a subset of $C^*$ that is such an intersection.
For each node $u$, let $a(u)$ be an optimal center in $C^*$
that is at most $\dist(e_i)$ away from $u$.
Consider the set $\calS^*=\{a(u)\}_{u\in C}$.
Since $\calS^*$ is a subset of $C^*$,
it is an independent set of $\calM$ by the definition of matroid.
It is also easy to see that $a(u)\in \B(u)$ for each $u\in C$.
Therefore, $\calS^*$ is also independent in $\calM_\B$, which proves our claim.
Thus, the algorithm returns a set $\calS$ that contains exactly 1 element from each $\B(v)$ with $v\in C$.
According to the algorithm, for each $v\in V$ there exists $u \in C$ that is at most $2 \dist(e_i)$ away, and this $u$ is within distance $\dist(e_i)$ from the (unique) element in $\B(u) \cap \calS$.
Thus every node of $V$ is
within a distance $3 \dist(e_i)=3r^*$ from some center in $\calS$.
\end{proof}

\subsection{Dealing with Outliers: \rmcenter}
We now consider the outlier version of \mcenter, denoted as \rmcenter, in which
an additional parameter $p$ is given and
the goal is to place centers (which must form an independent set)
such that after excluding at most $|V|-p$ nodes as outliers,
the maximum service cost of any node is minimized.
For $p = |V|$,  we have the standard \mcenter.
In this section, we present a $7$-approximation for \rmcenter.

Our algorithm bears some similarity to the 3-approximation algorithm for
robust $k$-center by Charikar et al. \cite{charikar2001algorithms}, who
also showed that robust $k$-center with forbidden centers
cannot be approximated within $3-\epsilon$ unless P = NP.
However, their algorithm for robust $k$-center
does not directly yield any approximation ratio for the forbidden center version.
In fact, robust $k$-center with forbidden centers is a special case of \rmcenter\
since forbidden centers can be easily captured by a partition matroid.
We briefly describe the algorithm in \cite{charikar2001algorithms}. Assume we have guessed the right optimal radius
$r$. For each $v \in V$, call $\B(v, r)$ the {\it disk} of $v$
and $\B(v, 3r)$ the {\it expanded disk} of $v$.
Repeat the following step $k$ times:
Pick an uncovered vertex as a center such that its disk covers the most number of uncovered nodes, then
mark all nodes in the corresponding expanded disk as covered.
Using a clever charging argument
they showed that at least $p$ nodes can be covered, which gives a $3$-approximation.
However, their algorithm and analysis do not extend to our problem in a straightforward manner.
The reason is that even if at least $p$ nodes are covered, we may not be able to find enough centers
in the covered region due to the matroid constraint.
In order to remedy this issue, we need to search for centers in the entire graph, which
also necessitates a more careful charging argument to show that we can cover at least $p$ nodes.

\begin{algorithm}[t]
  \caption{Algorithm for \rmcenter\ on $G_i$} \label{alg:robust-mcenter}
  Initially, set $C\leftarrow \emptyset$ and mark all vertices in $V$ as uncovered.

  \While{$V$ contains uncovered vertices}
  {
    Pick an uncovered vertex $v$ such that $\B(v, \dist(e_i))$ covers the most number of uncovered elements.

    $\B(v)\leftarrow\B(v, \dist(e_i))$. ($\B(v)$ is called the disk of $v$.)

    $\E(v)\leftarrow\B(v, 3\dist(e_i))\setminus \cup_{u\in C}\E(u)$. ($\E(v)$ is called the expanded disk of $v$. This definition ensures that all expanded disks in $\{\E(u)\}_{u\in C}$ are pairwise disjoint.)

    $C\leftarrow C\cup \{v\}.$
    Mark all vertices in $\E(v)$ as covered.
  }

  Create a set $U$ of (vertex, expanded disk) pairs, as follows:
  For each $v\in V$ and $u\in C$, if $\B(v, \dist(e_i))\cap \B(u, 3\dist(e_i))\ne \emptyset$,
  we add $(v, \E(u))$ to $U$. The weight $w((v, \E(u)))$ of the pair $(v, \E(u))$ is $|\E(u)|$.

  Define two matroids $\calM_1$ and $\calM_2$ over $U$
  as follows:
  { \begin{itemize}
  \item A subset $\{(v_i, \E(u_i))\}$ is independent
  in $\calM_1$ if all
    $v_i$'s in the subset are \newline distinct and form an independent
    set in $\calM$.
  \item  A subset $\{(v_i, \E(u_i))\}$ is independent in $\calM_2$ if
  all $\E(u_i)$'s in the subset are distinct. \newline (It is easy to see
  $\calM_2$ is a partition matroid.) \end{itemize} }

  Solve the matroid intersection problem between $\calM_1$ and
  $\calM_2$ optimally (note that the independence oracles for $\calM_1$ and $\calM_2$ can be easily simulated in polynomial time). Let $\calS$ be an optimal intersection. If
  $w(\calS)<p$, then we declare a failure and try the next $G_i$.
  Otherwise, we succeed and return $V(\calS)$ as the set of centers, where
  $V(\calS)=\{v\mid (v, \E(u))\in \calS \textrm{~for some~}u \in C \}$.
\end{algorithm}

Now we describe our algorithm and prove its performance guarantee.
For each $1 \le i \le {|V|\choose 2}$, we run Algorithm
\ref{alg:robust-mcenter} on the graph $G_i$ defined as before.
We need the following simple
lemma.

\begin{lemma}
\label{lm:m1}
$\calM_1$ is a matroid.
\end{lemma}
\begin{proof}
It is straightforward to verify that the first and second matroid properties hold.
We only need to verify the third property.
Suppose $A$ and $B$ are two independent sets of $\calM_1$ and $|A|>|B|$.
We know the set $V(A)$ (resp., $V(B)$) of vertices that appear in $A$ (resp.,
$B$) is an independent set of $\calM$.
Since $|V(A)|=|A|$ and $|V(B)|=|B|$, $|V(A)|>|V(B)|$. Hence, there is a vertex $v\in V(A)\setminus V(B)$
such that $V(B)\cup \{v\}$ is independent.
We add to $B$ the pair in $A$ that involves $v$ and it is easy to see the resulting set is also independent in $\calM_1$.
\end{proof}

\begin{theorem}
Algorithm \ref{alg:robust-mcenter} produces a $7$-approximation for
\rmcenter.
\end{theorem}

\begin{proof}
Assume the maximum radius of any cluster in an optimal solution is $r^*$
and the set of optimal centers is $C^*$.
For each $v\in C^*$, let $\O(v)$ denote the optimal disk $\B(v, r^*)$.
As before, we claim that our algorithm succeeds if $\dist(e_i) = r^*$.
It suffices to show the existence of an intersection of $\calM_1$ and $\calM_2$
with a weight at least $p$.
We next construct such an intersection $\calS'$ from the optimal center set $C^*$.
The high level idea is as follows.
Let the disk centers in $C$ be
$v_1, v_2, \ldots, v_k$ (according to the order that our algorithm chooses them).
Note that $v_1, v_2, \ldots, v_k$ are the centers chosen by the greedy procedure in the first part of the algorithm, but not the centers returned at last.
We process these centers one by one.
Initially, $\calS'$ is empty.
As we process a new center $v_j$, we may add $(v, \E(v_j))$ for some $v\in C^*$ to $\calS'$.
Moreover, we charge each newly covered node
in any optimal disk
to some nearby node in the expanded disk $\E(v_j)$.
(Note that this is the key difference between our charging argument
and that of \cite{charikar2001algorithms};
in \cite{charikar2001algorithms}, a node may be charged to some node far away.)
We maintain that all nodes in $\cup_{v\in C^*} \O(v)$ covered by $\cup_{j'=1}^j \E(v_{j'})$
are charged after processing $v_j$.
Thus, eventually,  all nodes covered by the optimal solution (i.e., $\cup_{v\in C^*} \O(v)$) are charged to the expanded disks selected by our algorithm.
We also make sure that each node in any expanded disk in $\calS'$ is being charged to at most once.
Therefore, the weight of $\calS'$ is at least $|\cup_{v\in C^*} \O(v)|\geq p$.

Now, we present the details of the construction of $\calS'$.
If every node in $\O(v)$ for some $v\in C^*$ is charged, we say $\O(v)$ is {\em entirely charged}.
Consider the step when we process $v_j \in C$.
We distinguish the following cases.

\begin{enumerate}
\item
Suppose there is a node $v\in C^*$ such that
$\O(v)$ is not entirely charged and $\O(v)$ intersects $\B(v_j)$.
Then add $(v, \E(v_j))$ to $\calS'$
(if there are multiple such $v$'s, we only add one of them).
We charge the newly covered nodes in $\cup_{v\in C^*}\O(v)$
(i.e., the nodes in $(\cup_{v\in C^*}\O(v))\cap \E(v_j)$)
 to themselves
(we call this charging rule I).
Note that $\O(v)$ is entirely charged  after this step since
$\O(v)\subseteq \B(v_j, 3r^*)$.

\item Suppose $\B(v_j)$ does not intersect
$\O(v)$ for any $v\in C^*$, but there is some node $v\in C^*$ such
that $\O(v)$ is not entirely charged and $\O(v)\cap \E(v_j)\ne
\emptyset$. Then we add $(v, \E(v_j))$ to $\I$ and charge all newly
covered nodes in $\O(v)$ (i.e., the node in $\O(v)\cap \E(v_j)$) to
$\B(v_j)$ (we call this charging rule II). Since $\B(v_j)$ covers
the most number of uncovered elements when $v_j$ is added, there are
enough vertices in $\B(v_j)$ to charge. Obviously, $\O(v)$ is
entirely charged after this step. If there is some other node $u\in
C^*$ such that $\O(u)$ is not entirely charged and $\O(u)\cap
\E(v_j)\ne \emptyset$, then we charge each newly covered node (i.e.,
nodes in $\O(u)\cap \E(v_j)$) in $\O(u)$ to itself using
rule I.

\item If $\E(v_j)$ does not intersect with any optimal disk $\O(v)$
that is not entirely charged, then we simply
skip this iteration and continue to the next $v_j$.
\end{enumerate}

It is easy to see that all covered nodes in $\cup_{v\in C^*} \O(v)$
are charged in the process and each node is being charged to at most once.
Indeed, consider a node $u$ in $\B(v_j)$.
If $\B(v_j)$ intersects some $\O(v)$, then $u$ may be charged by rule I
and, in this case, no further node can be charged to $u$ again.
If $\B(v_j)$ does not intersect any $\O(v)$, then $u$ may be charged by rule II.
This also happens at most once. It is
obvious that in this case, no node can be charged to $u$ using rule I.
For a node $u \in \E(v_j)\setminus\B(v_j)$, it can be charged at most once
using rule I.
Moreover, by the charging process,
all nodes in $\cup_{v\in C^*} \O(v)$
are charged to the nodes in some expanded disks that appear in $\calS'$.
Therefore, the total weight of $\calS$ is at least $p$.
We can see that each vertex in $V(\calS')$ is also in $C^*$
and appears at most one. Therefore, $\calS'$ is independent in $\calM_1$.
Clearly, each $\E(u)$ appears in $\calS'$ at most once. Hence,
$\calS'$ is also independent in $\calM_2$,
which proves our claim.

Since $\calS$ is an optimal intersection, we know the expanded disks in $\calS$
contain at least $p$ nodes.
By the requirement of $\calM_1$, we can guarantee that
the set of centers forms an independent set in $\calM$.
For each $(v, \E(u))$ in $\calS$, we can see that every node $v'$ in $\E(u)$ is
within a distance $7\dist(e_i)$ from $v$, as follows.
Suppose $u'\in \B(v, \dist(e_i))\cup \B(u, 3\dist(e_i))$ (because $\B(v, \dist(e_i))\cup \B(u, 3\dist(e_i))\ne \emptyset$ for any pair $(v,\E(u))\in U$).
By the triangle inequality, $\dist(v',v)\leq \dist(v',u)+\dist(u,u')+\dist(u',v)
\leq 3\dist(e_i)+3\dist(e_i)+\dist(e_i)=7\dist(e_i).$
This completes the proof of the theorem.
\end{proof}

\section{The Knapsack Center Problem}

In this section, we study the \multiknap\ problem and its outlier version. Recall that an input of \multiknap\ consists of a metric space $(\V, d)$, $m$ nonnegative weight functions $w_1,
\ldots, w_m$ on $\V$, and $m$ budgets
$\kb_1,\ldots,\kb_m$.
The goal is to select a
set of centers $\calS \subseteq \V$ with
$w_i(\calS) \leq \kb_i$ for all $1\leq i\leq m$, so as to minimize the maximum service cost of any vertex in $\V$.
In the outlier version of \multiknap, we are given an additional parameter $p \leq |V|$, and the objective is to minimize $cost_p(\mathcal{S}) := \min_{\V' \subseteq \V: |\V'|\geq p}\max_{v
\in \V'}\min_{i\in \calS}d(v,i)$, i.e., the maximum service cost of any non-outlier node
after excluding at most $|V|-p$ nodes as outliers.

\subsection{Approximability of \multiknap}
When there is only one knapsack constraint (i.e., $m=1$), the
problem degenerates to the weighted $k$-center problem for which a
3-approximation algorithm exists \cite{jacm86}. However, as we show in Theorem \ref{thm:hard_multi_knap}, the
situation changes dramatically
even if there are only two knapsack constraints.

\begin{theorem}\label{thm:hard_multi_knap}
For any $f>0$, if there is an $f$-approximation algorithm for \multiknap with two knapsack constraints,
then $\P=\NP$.
\end{theorem}
\begin{proof}
To prove the theorem, we present a reduction from the partition
problem, which is well-known to be NP-hard \cite{book_npc}, to the
\multiknap problem with two knapsack constraints. In the partition problem, we are given a
multiset of positive integers $\mathcal{S}=\{s_1, s_2, \ldots,
s_n\}$, and the goal is to decide whether $\calS$ can be partitioned
into two subsets such that the sum of numbers in one subset equals
the sum of numbers in the other subset.

Given an instance $\calS=\{s_1,s_2,\ldots,s_n\}$ of the partition
problem, we construct an instance $\calI$ of the \multiknap problem
as follows. The set of clients is $\V = \{a_i,b_i~|~1\leq i\leq
n\}$. The distance metric $d$ is defined as $d(a_i,b_i)=0$ for all
$1\leq i\leq n$, and $d(a_i,a_j)=d(a_i,b_j)=d(b_i,b_j)=1$ for all
$i\neq j$. It is easy to verify that $d$ is indeed a metric. Every
client in $\V$ has a unit demand. There are two weight functions $w_1$
and $w_2$ specified as follows: for each $1\leq i\leq n$,
$w_1(a_i)=s_i$, $w_1(b_i)=0$, $w_2(a_i)=0$, and $w_2(b_i)=s_i$.
The two corresponding weight budgets are $\kb_1=\kb_2=T/2$, where
$T=\sum_{j=1}^{n}s_j$. This finishes the construction of $\calI$.

We show that $\calS$ can be partitioned into two subsets of equal
sum if and only if $\calI$ has a solution of cost 0. First consider
the ``if'' direction. Assume that $\calI$ admits a solution of cost
0. Clearly, for each $1\leq i\leq n$, the solution must take at
least one of $\{a_i,b_i\}$ as a center, and we assume w.l.o.g. that
it takes exactly one of $a_i$ and $b_i$ (just choosing an arbitrary one
if both are taken). Let $I_1$ be the set of indices $i$ for which
$a_i$ is taken as a center in the solution. Then
$I_2=\{1,2,\ldots,n\} \setminus I_1$ consists of all indices $i$ for
which $b_i$ is taken by the solution. Considering the first weight
constraint, we have $T/2=\kb_1 \geq \sum_{i\in
I_1}w_1(a_i)+\sum_{i\in I_2}w_1(b_i)=\sum_{i\in I_1}s_i$. Similarly,
by the second weight constraint, we get $T/2 \geq \sum_{i\in
I_2}s_i$. Since $\sum_{i\in I_1}s_i+\sum_{i\in
I_2}s_i=\sum_{i=1}^{n}s_i=T$, it holds that $\sum_{i\in
I_1}s_i=\sum_{i\in I_2}s_i=T/2$. Therefore, $\calS$ can be
partitioned into two subsets of equal sum.

We next prove the ``only if'' part. Suppose there exists $I_1
\subseteq \{1,2,\ldots,n\}$ such that $\sum_{i\in I_1}s_i=T/2$. In
the instance $\calI$, we take $\calT:=\{a_i~|~i\in
I_1\}\cup\{b_j~|~j\in \{1,2,\ldots,n\}\setminus I_1\}$ as the set of
centers. It only remains to show that $\calT$ satisfies both the weight
constraints, which is easy to verify: $\sum_{v\in
\calT}w_1(v)=\sum_{i\in I_1}s_i=T/2\leq \kb_1$, and $\sum_{v\in
\calT}w_2(v)=\sum_{j\in \{1,2,\ldots,n\}\setminus
I_1}s_j=T-\sum_{j\in I_1}s_j=T/2\leq \kb_2$. This proves the ``only
if'' direction.

Since the optimal objective value of $\calI$ is 0, any
$f$-approximate solution is in fact an optimal one. Hence, if
\multiknap with two constraints and unit demands allows an
$f$-approximation algorithm for any $f>0$, then the partition
problem can be solved in polynomial time, which implies $\P = \NP$.
The proof of Theorem~\ref{thm:hard_multi_knap} is thus complete.
\end{proof}

It is then natural to ask whether a
constant factor approximation can be obtained if the constraints can be relaxed slightly. We show in Theorem~\ref{thm:apx_multiknap} that this is achievable (even for the demand version).
Before proving the theorem we first present some high-level ideas of our algorithm, shown as Algorithm~\ref{alg:multi_knap}. The algorithm first guesses the optimal cost $\opt$, and then chooses a collection of disjoint disks of radius $\opt$ according to some rules. It can be shown that there exists a set of centers consisting of exactly one point from each disk that gives a 3-approximate solution and satisfies all the knapsack constraints.
We then reduce the remaining task to another problem called the \emph{\groupknap problem}, which will formally be defined in the following proof.

\begin{theorem}\label{thm:apx_multiknap}
For any fixed $\epsilon>0$, there is a 3-approximation algorithm for \multiknap\ with a constant number of knapsack constraints, which is guaranteed to satisfy one constraint and
violate each of the others by at most a factor of
$1+\epsilon$.
\end{theorem}

In what follows we prove Theorem~\ref{thm:apx_multiknap}.
We first present our algorithm for \multiknap\ in Algorithm~\ref{alg:multi_knap} that we use to prove Theorem \ref{thm:apx_multiknap}. The algorithm works for the more general version where each vertex $v$ has a demand $r(v)$ and the service cost of $v$ is $\min_{i\in \calS}r(v)d(v,i)$ when taking $\calS$ as the set of centers.

\begin{algorithm}[t]
  \caption{Algorithm for \multiknap with multiple constraints} \label{alg:multi_knap}
  Guess the optimal objective value $\opt$.

  For each client $v\in \V$, let $\B(v) \leftarrow \B(v, \opt)$ be the \emph{disk} of $v$. Let $\calT \leftarrow \emptyset$.

  \While{there exists $i \in \V$ such that $\B(i) \cap \B(j) = \emptyset$ for all $j\in \calT$}
  {
  Choose such an $i$ with maximum demand, and let $\calT \leftarrow \calT \cup \{i\}$.
  }

  Create an instance $\calI$ of the \groupknap problem as $\calI=(\{\B(i)\}_{i\in \calT}, \{w_j,\kb_j\}_{1\leq j\leq m})$ (recall that $m=O(1)$), and get a solution $\calS$ by applying the algorithm indicated by Lemma \ref{lem:group}.

  \Return $\calS$
\end{algorithm}

Given an instance of the \multiknap problem, suppose Algorithm \ref{alg:multi_knap} correctly guesses the optimal objective value $\opt$. (This can be equivalently realized by
running the algorithm for all ${|V| \choose 2}$ possibilities and taking the best solution among all the candidates.)
The algorithm greedily finds a collection of mutually
disjoint disks $\{\B(i)\}_{i\in\calT}$, and then constructs a set of centers
by selecting exactly one point from each disk using some algorithm for the \groupknap problem, which we will define later.

Call a set $\calS \subseteq \V$ \emph{standard} if $\calS$ consists
of exactly one point from each of the disks $\{\B(i)\}_{i\in\calT}$.
We first show that there exists a standard set $\calS$ such that
$w_j(\calS)\leq \kb_j$ for all $1\leq j\leq m$, i.e., $\calS$
fulfills all the knapsack constraints. Suppose $\calO\subseteq \V$
is the set of centers opened in some optimal solution. Then, for
each $i\in\calT$, there exists $j\in \calO$ such that $r(i)
d(i,j)\leq \opt$, and thus $j\in \B(i)$. Hence, we can choose from
each $\B(i)$ exactly one point that belongs to $\calO$, and these
points are distinct because the disks are pairwise disjoint. Let
$\calS$ denote the set of these points. Clearly, $\calS$ is a standard
and is a subset of $\calO$, and thus $w_j(\calS)\leq w_j(\calO)\leq
\kb_j$ for all $1\leq j\leq m$. This proves the existence of a
standard set that satisfies all the knapsack constraints.

We will reduce the remaining task to another problem called the \groupknap problem, which we define as follows. Suppose we are given a collection of pairwise
disjoint sets $\{\calS_i\}_{1\leq i\leq n}$. Let $\calS
=\bigcup_{i=1}^{n}\calS_i$. For some fixed integer $m\geq 1$, there
are $m$ nonnegative weight functions defined on the items of $\calS$,
which we denote by $w_1, \ldots, w_m$, and $m$ weight limits
$\kb_1,\ldots,\kb_m$. A \emph{solution} is a subset $\calS'
\subseteq \calS$ that consists of exactly one element from each of
the $n$ sets $\calS_1, \ldots, \calS_n$. The goal is find a solution
$\calS'$ such that $w_j(\calS')\leq \kb_j$ for all $1\leq j\leq m$,
provided that such solution exists. For our purpose, we
require the number of constraints to be a constant. This problem is
new to our knowledge, and may be useful in other applications. By Lemma~\ref{lem:group} (which will be presented and proved later), we
can find in polynomial time a solution that satisfies one constraint
and violates each of the others by a small factor.

Now come back to the \multiknap\ problem.
By Lemma
\ref{lem:group}, line 6 of Algorithm \ref{alg:multi_knap} produces
in polynomial time a standard set $\calS$ that satisfies one
constraint and violates each of the others by a factor of at most $1+\epsilon$.
(We notice that, when running Algorithm \ref{alg:multi_knap}
with an incorrect value of $\opt$, there may not exist any standard
set, in which case the algorithm may return an empty set. We shall
simply ignore such solutions.)

It now only remains to show that, by designating $\calS$ as the set
of centers, the maximum service cost of any client is at most
$3\cdot \opt$.
Suppose $\calS \cap \B(i) = \{t_i\}$ for each $i\in\calT$. It
suffices to prove that, for each $j\in \V$, there exists $i\in\calT$
such that $r(j) d(j, t_i)\leq 3\cdot \opt$. We consider two cases.

\begin{enumerate}
\item $j \in \calT$. Since $t_j \in \B(j)$, we
have $r(j) d(j,t_j)\leq \opt \leq 3\opt$ by the definition of
$\B(j)$.

\item $j \not\in \calT$. Then $\B(j) \cap \B(i) \neq
\emptyset$ for some $i\in\calT$, otherwise $j$ should be added to
$\calT$ by the algorithm. Let $\mathcal{Q}=\{i\in\calT~|~\B(i) \cap
\B(j) \neq \emptyset\}$. If $r(i)<r(j)$ for all $i\in \mathcal{Q}$, then
the algorithm will choose $j$ before choosing all $i\in
\mathcal{Q}$, which contradicts with the assumption that $j\not\in
\calT$. Thus, there exists $i\in \mathcal{Q}$ for which $r(i)\geq
r(j)$. Consider this particular $i$, and choose an arbitrary $i' \in
\B(i) \cap \B(j)$. We have
\begin{eqnarray*}
r(j) d(j, t_i) &\leq& r(j)(d(j,i')+d(i,i')+d(i,t_i)) \textrm{~~~~~~~~(triangle inequality)}\\
&\leq& r(j)d(j,i') + r(i)d(i,i') + r(i)d(i,t_i) \textrm{~~(because~}r(i)\geq r(j))\\
&\leq& \opt + \opt + \opt \textrm{~~~(due to the definition of disks)} \\
&=& 3\cdot \opt.
\end{eqnarray*}
\end{enumerate}

Combining the two cases, we have shown that the service cost with centers in $\calS$ is at
most three times the optimal cost, which completes the proof.

Finally, we need the following Lemma~\ref{lem:group}, which is used in the above argument. The \groupknap\ problem is similar to the multiple knapsack problem (i.e., the knapsack problem with multiple resource constraints), and the (standard) technique for the latter can be easily adapted to solve the \groupknap\ problem (see, e.g., \cite{phdthesis/Pisinger95,books/knap04}). Another way to deduce Lemma~\ref{lem:group} is by applying the $\epsilon$-approximate Pareto curve method introduced by Papadimitriou and Yannakakis \cite{conf/focs/PapadimitriouY00}. For sake of completeness, we give a proof of Lemma~\ref{lem:group} in Appendix~\ref{apx:group}.

\begin{lemma}\label{lem:group}
For any fixed $\epsilon>0$, there is a polynomial time algorithm
that, given an instance of \groupknap for which a solution
satisfying all weight constraints exists, constructs in polynomial
time a solution that satisfies one constraint and violates each of
the others by at most a factor of $1+\epsilon$.
\end{lemma}

\begin{algorithm}[t]
  \caption{Algorithm for \rknap} \label{alg:robust_knap}
  Guess the optimal objective value $\opt$.

  For each $v\in \V$, let $\B(v) \leftarrow \B(v,\opt)$ and
  $\E(v) \leftarrow \B(v,3\opt)$.

  $\mathcal{S} \leftarrow \emptyset; \mathcal{C} \leftarrow \emptyset$ (the points in $\mathcal{C}$ are {\em covered} and those in
  $\V \setminus \mathcal{C}$ are {\em uncovered}).

  \While{$w(\mathcal{S})<\kb$ and $\V \setminus \mathcal{C} \neq \emptyset$
  }
  {
  Choose $i \in \V \setminus \mathcal{S}$ that maximizes $\frac{|\B(i) \setminus \mathcal{C}|}{w(i)}$.

  $\mathcal{S} \leftarrow \mathcal{S} \cup \{i\}; \mathcal{C} \leftarrow \mathcal{C} \cup \E(i)$ (i.e., mark all uncovered points in $\E(i)$ as covered).
  }

  \Return $\mathcal{S}$
\end{algorithm}

\subsection{Dealing with Outliers: \rknap}\label{sec:robust_knapsack}
We now study \rknap, the outlier version of \multiknap. Here we consider the case with one knapsack constraint (with weight function $w$ and budget $\kb$) and unit demand. Our main theorem is as follows.

\begin{theorem}\label{thm:robust_apx}
There is a 3-approximation algorithm for \rknap that violates the
knapsack constraint by at most a factor of $1+\epsilon$ for any
fixed $\epsilon>0$.
\end{theorem}

We present our algorithm for \rknap as Algorithm
\ref{alg:robust_knap}. We assume that $\kb<w(\V)$, since otherwise the problem is trivial.
We also set $A/0:=\infty$ for $A>0$ and $0/0:=0$, which makes line 5 work even if $w(i)=0$.
Our algorithm can be regarded as a ``weighted'' version of that of Charikar et al. \cite{charikar2001algorithms}, but the analysis is much more involved.
We next prove the following theorem, which can be used together with the partial enumeration technique to yield Theorem \ref{thm:robust_apx}.
Note that, if all clients have unit weight, Theorem \ref{thm:robust_additive} will guarantee a 3-approximate solution $\calS$ with $w(\calS)<\kb+1$, which implies $w(\calS)\leq \kb$. So it actually gives a 3-approximation without violating the constraint. Thus, our result generalizes that of Charikar et al. \cite{charikar2001algorithms}.

\begin{theorem}\label{thm:robust_additive}
Given an input of the \rknap problem, Algorithm
\ref{alg:robust_knap} returns a set $\mathcal{S}$ with
$w(\mathcal{S})<\kb+\max_{v\in \V}w(v)$ such that $cost_p(\mathcal{S})
\leq 3\opt$.
\end{theorem}

\begin{proof}
We call $\B(v)$
the disk of $v$ and $\E(v)$ the expanded disk of $v$.
Assume w.l.o.g. that the algorithm returns
$\calS=\{1,2,\ldots,q\}$ where $q=|\calS|$, and that the centers are
chosen in the order $1, 2, \ldots, q$.
We first observe that $\B(1),\ldots,\B({q})$ are pairwise
disjoint, which can be seen as follows.
By standard use of
the triangle inequality, we have $\B(i) \subseteq \E(j)$ and $\B(j)
\subseteq \E(i)$ for any $i,j\in \V$ such that $\B(i) \cap \B(j)
\neq \emptyset$. Therefore, if there exists $1\leq i<j\leq q$ such
that $\B(j) \cap \B(i) \neq \emptyset$, then all points in $\B(j)$
are marked ``covered'' when choosing $i$, and hence choosing $j$
cannot cover any more point, contradicting with the way in which the centers
are chosen (note that the algorithm terminates when all points have
been covered). So the $q$ disks $\B(1),\ldots,\B({q})$ are pairwise
disjoint.

For ease of notation, let $\B({\V'}) := \bigcup_{v\in \V'}\B(v)$
and $\E(\V'):=\bigcup_{v\in \V'}\E(v)$ for $\V' \subseteq \V$.
By the condition of the WHILE loop,
$w(\{1,\ldots,q-1\})<\kb$, and thus $w(\mathcal{S})<\kb+w(q)\leq
\kb+\max_{v\in \V}w(v)$. It remains to prove
$cost_p(\mathcal{S})\leq 3\opt$. Note that this clearly holds if the
expanded disks $\E(1), \ldots, \E({q})$ together cover at least $p$
points.
Thus, it suffices to show that $|\E(\mathcal{S})|\geq p$. If
$w(S)<\kb$, then all points in $\V$ are covered by $\E({\calS})$
due to the termination condition of the WHILE loop, and thus
$|\E({\calS})|=|\V|\geq p$. In the rest of the proof, we deal with
the case $w(\calS)\geq \kb$.

For each $v \in \V$, let $f(v)$ be the minimum $i\in \mathcal{S}$
such that $\B(v) \cap \B(i) \neq \emptyset$; let $f(v)=\infty$ if no
such $i$ exists (i.e., if disk $\B(v)$ is disjoint from all disks
centered in $\mathcal{S}$).
Suppose
$\mathcal{O}=\{o_1,o_2,\ldots,o_m\}$ is an optimal solution, in
which the centers are ordered such that $f(o_1)\leq
\cdots \leq f(o_m)$.
Since the optimal solution is also feasible, we have $|\B({\mathcal{O}})|\geq p$.
Hence, to prove $|\E({\mathcal{S}})|\geq p$, we only need to show $|\E({\mathcal{S}})| \geq |\B({\mathcal{O}})|$. For any sets
$A$ and $B$, we have $|A| = |A\setminus B| + |A \cap B|$. Therefore,
\begin{eqnarray}
& &|\E({\mathcal{S}})|-|\B({\mathcal{O}})| \notag\\
&=& (|\E({\mathcal{S}})\setminus \B({\mathcal{O}})|+|\E({\mathcal{S}}) \cap \B({\mathcal{O}})|) - (|\B({\mathcal{O}})\setminus \E({\mathcal{S}})|+|\E({\mathcal{S}}) \cap \B({\mathcal{O}})|)\notag\\
&=& |\E({\mathcal{S}}) \setminus \B({\mathcal{O}})| - |\B({\mathcal{O}}) \setminus \E({\mathcal{S}})| \notag\\
&\geq& |\B({\mathcal{S}}) \setminus \B({\mathcal{O}})| - |\B({\mathcal{O}}) \setminus \E({\mathcal{S}})| \textrm{~~~(because~}\B({\mathcal{S}}) \subseteq \E({\mathcal{S}})).
\end{eqnarray}
As $\B(1),\ldots,\B(q)$ are pairwise disjoint,
$$|\B({\mathcal{S}}) \setminus \B({\mathcal{O}})|=|\cup_{i\in \mathcal{S}}(\B(i) \setminus \B({\mathcal{O}}))|=\sum_{i\in \mathcal{S}}|\B(i) \setminus \B({\mathcal{O}})|,$$ and
$$|\B({\mathcal{O}}) \setminus \E({\mathcal{S}})|=|\cup_{j=1}^{m}(\B({o_j}) \setminus \E({\mathcal{S}}))|\leq \sum_{j=1}^{m}|\B({o_j}) \setminus \E({\mathcal{S}})|.$$ Thus,
\begin{equation}\label{equ:goal}
|\E({\mathcal{S}})|-|\B({\mathcal{O}})| \geq \sum_{i\in \mathcal{S}}|\B(i) \setminus \B({\mathcal{O}})| - \sum_{j=1}^{m}|\B(o_j) \setminus \E(\mathcal{S})|.
\end{equation}

{ \centering
\begin{tikzpicture}[scale=2/3]
\coordinate[label=right:$1$] (A1) at (0, 0);
\coordinate[label=right:$2$] (A2) at (3.5, 0);
\coordinate[label=right:$3$] (A3) at (5.7, 0.7); \fill (A1) circle
(1pt); \draw (A1) circle (1); \fill (A2) circle (1pt); \draw (A2)
circle (1); \fill (A3) circle (1pt); \draw (A3) circle (1); \draw
(A1) [dashed] circle (3); \draw (A2) [dashed] circle (3); \draw (A3)
[dashed] circle (3); \node at (-2.5,1.0) {$\mathsf{E}_1$}; \node at
(2.9,2.9) {$\mathsf{E}_2$}; \node at (6.4,3.2) {$\mathsf{E}_3$};

\coordinate[label=right:$o_1$] (O1) at (-1, -1.5);
\coordinate[label=right:$o_2$] (O2) at (0.6, -1.7);
\coordinate[label=right:$o_3$] (O3) at (2.6, -1);
\coordinate[label=right:$o_4$] (O4) at (4.8, -0.8);
\coordinate[label=right:$o_5$] (O5) at (6.5, -2.2);
\coordinate[label=right:$o_6$] (O6) at (8.4, -1.2);

\fill (O1) circle (1pt); \draw (O1) circle (1); \fill (O2) circle
(1pt); \draw (O2) circle (1); \fill (O3) circle (1pt); \draw (O3)
circle (1); \fill (O4) circle (1pt); \draw (O4) circle (1); \fill
(O5) circle (1pt); \draw (O5) circle (1); \fill (O6) circle (1pt);
\draw (O6) circle (1);

\end{tikzpicture}

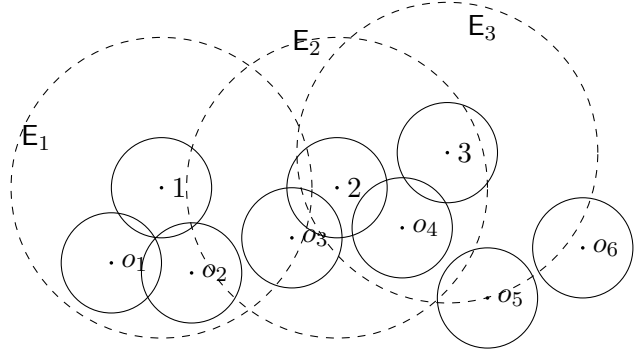
\captionof{figure}{\small An example of the algorithm for {\sf
Robust-KnapCenter}. The algorithm returns $\mathcal{S}=\{1,2,3\}$,
and the optimal solution opens $\{o_1,o_2,\ldots,o_6\}$. Disks and
extended disks are represented by (small) circles and (large) dashed
circles, respectively. In this case, we have $f(o_1)=f(o_2)=1,
f(o_3)=f(o_4)=2$, $f(o_5)=f(o_6)=\infty$, and thus $t=5$. Then,
$\mathsf{R}(1)=\{1,2\}, \mathsf{R}(2)=\{3,4\}$, and
$\mathsf{R}(3)=\emptyset$. \normalsize}\label{fig:robustknap}
\vspace{4mm}
}

Let $t$ be the unique integer in $\{1,\ldots,m+1\}$ such
that $f(o_j)\leq |\mathcal{S}|$ for all $1\leq j\leq t-1$ and
$f(o_j)=\infty$ for all $t\leq j\leq m$.
(That is, each disk $\B(o_j)$
($1\leq j\leq t-1$) intersects with $\B(i)$ for some $i\in
\mathcal{S}$, while the remaining $\B(o_{t}),\ldots, \B(o_m)$ are
disjoint from all the disks of points in $\mathcal{S}$. Such $t$
exists because $f(o_1)\leq \cdots \leq f(o_m)$. See Figure
\ref{fig:robustknap} for an example.)
Then,
for all $1\leq j\leq t-1$, we have $\B(o_j) \cap \B(f(o_j)) \neq
\emptyset$, and thus $\B(o_j) \subseteq \E(f(o_j)) \subseteq
\E(\mathcal{S})$, implying that $|\B(o_j) \setminus
\E(\mathcal{S})|=0$ for all $1\leq j\leq t-1$. Combining with the
inequality (\ref{equ:goal}), we have
\begin{equation}\label{equ:goal2}
|\E(\mathcal{S})|-|\B(\mathcal{O})| \geq \sum_{i\in \mathcal{S}}|\B(i)\setminus \B(\mathcal{O})| - \sum_{j=t}^{m}|\B(o_j) \setminus \E(\mathcal{S})|.
\end{equation}
Hence, it suffices to prove that
\begin{equation}\label{equ:goal3}
\sum_{i\in \mathcal{S}}|\B(i) \setminus \B(\mathcal{O})| - \sum_{j=t}^{m}|\B(o_j) \setminus \E(\mathcal{S})| \geq 0.
\end{equation}

The inequality is trivial when $t=m+1$. Thus, we assume in what follows that $t\leq m$, i.e., $\B(o_m)$ is disjoint from $\B(1),\B(2),\ldots,\B(q)$.
Before proving (\ref{equ:goal3}), we introduce some notations.
For each $i \in \mathcal{S}$, define
$\mathsf{R}(i) := \{j~|~1\leq j\leq m; f(o_j)=i\},$
and let $l(i):=\min\{j \mid j\in \mathsf{R}(i)\}$ and
$q(i):= \max \{j \mid j\in \mathsf{R}(i)\}$
be the minimum index and maximum index in $\mathsf{R}(i)$, respectively
(let $l(i)=q(i)=\infty$ if $\mathsf{R}(i)=\emptyset$).
By the definitions of $f(\cdot)$ and $t$, each $\mathsf{R}(i)$ is a set of consecutive integers (or empty), and $\{\mathsf{R}(i)\}_{i\in \mathcal{S}}$ forms a partition of $\{1,2,\ldots,t-1\}$.
Also, $q(i)=l(i+1)-1$ if $l(i+1)\neq \infty$. See Figure \ref{fig:robustknap} for an illustration of the notations.

Consider an arbitrary $i\in \mathcal{S}$. For each $j$ such that $l(i+1)\leq j\leq t-1$, we know that $j \in \mathsf{R}(i')$ for some $i'>i$, i.e., $f(o_j)=i'>i$, and thus $\B(o_j) \cap \B(i) = \emptyset$. By the definition of $t$, we also have $\B(o_j) \cap \B(i) = \emptyset$ for all $t\leq j\leq m$. Therefore,
\begin{equation}\label{equ:disjoint}
\B(o_j) \cap \B(i) = \emptyset \textrm{~for all~}j \textrm{~s.t.~}\min\{t,l(i+1)\}\leq j\leq m.
\end{equation}
(Here we take the minimum of $l(i+1)$ and $t$ because $l(i+1)$ may be $\infty$.)

We next try to lower-bound $|\B(i) \setminus \B(\mathcal{O})|$ in order to establish (\ref{equ:goal3}).
Equality (\ref{equ:disjoint}) tells us that $\B(o_j) \cap \B(i) \neq \emptyset$ implies $j\in \mathsf{R}(1) \cup \cdots \cup \mathsf{R}(i)$.
In consequence,
\begin{equation}\label{equ:t1}
\B(i) \setminus \B(\mathcal{O}) = \B(i) \setminus \cup_{j=1}^{m}\B({o_j}) = \B(i) \setminus \cup_{j\in \mathsf{R}(1) \cup \cdots \cup \mathsf{R}(i)}\B({o_j}).
\end{equation}
For each $j \in \mathsf{R}({i'})$ with $1\leq i'\leq i-1$, $\B({o_j}) \cap \B({i'}) \neq \emptyset$, and thus $\B(o_j) \subseteq \E({i'}) \subseteq \E(\{1,2,\ldots,i-1\})$. For convenience, define $\E_{<i}:=\E({\{1,2,\ldots,i-1\}})$. Then, from (\ref{equ:t1}) we get $\B(i) \setminus \B({\mathcal{O}}) \supseteq \B(i) \setminus (\E_{<i} \cup \bigcup_{j\in \mathsf{R}(i)}\B({o_j}))$, and hence
\begin{eqnarray}\label{equ:t2}
& &|\B(i) \setminus \B({\mathcal{O}})| \geq |\B(i) \setminus (\E_{<i} \cup \bigcup_{j\in \mathsf{R}(i)}\B(o_j))| \notag\\
&=& |\B(i) \setminus (\E_{<i} \cup \bigcup_{j\in \mathsf{R}(i)}(\B({o_j}) \setminus \E_{<i}))| \textrm{~~(because}~ A \cup \bigcup_i B_i = A \cup \bigcup_i(B_i \setminus A)) \notag\\
&=&|(\B(i) \setminus \E_{<i}) \setminus \bigcup_{j\in \mathsf{R}(i)}(\B({o_j}) \setminus \E_{<i})|
\geq |\B(i) \setminus \E_{<i}| - \sum_{j\in \mathsf{R}(i)}|\B({o_j}) \setminus \E_{<i}|.
\end{eqnarray}

Now consider the particular execution of line 5 in which $i$ is chosen and added to $\calS$.
Note that (\ref{equ:disjoint}) holds for all $i\in \mathcal{S}$. Thus, for all $1\leq i' \leq i-1$ and $\min\{t,l(i'+1)\}\leq j\leq m$,
$\B({o_j})$ is disjoint from $\B({i'})$, which in particular implies $o_j
\not\in \B({i'})$. By considering all $i'\in\{1,\ldots,i-1\}$ and
noting that $l(i)\geq l(i'+1)$, we have $o_j \not\in
\B(\{1,2,\ldots,i-1\})$ for all $\min\{t,l(i)\}\leq j\leq m$. This
further indicates that $\{1,2,\ldots,i-1\} \cap
\{o_j~|~\min\{t,l(i)\}\leq j\leq m\} = \emptyset.$ Recall that
$1,2,\ldots,i-1$ are all the points added to $\mathcal{S}$ before
$i$. Therefore, no point in $\{o_j~|~\min\{t,l(i)\}\leq j\leq m\}$
was chosen before $i$. By our way of choosing centers (see
line 5), we have
\begin{equation}\label{equ:choose}
\frac{|\B(i) \setminus \E_{<i}|}{w(i)} \geq \frac{|\B({o_j}) \setminus \E_{<i}|}{w(o_j)} \textrm{~for all~}j \textrm{~s.t.~}\min\{t,l(i)\}\leq j\leq m.
\end{equation}
Hence, for all $j\in \mathsf{R}(i)$,
$$|\B({o_j}) \setminus \E_{<i}|\leq \frac{w(o_j)}{w(i)}|\B(i) \setminus \E_{<i}|.$$
Substituting the above inequality into (\ref{equ:t2}) gives
\begin{eqnarray}\label{equ:t3}
|\B(i) \setminus \B({\mathcal{O}})| &\geq& |\B(i) \setminus \E_{<i}| - \sum_{j\in \mathsf{R}(i)}\frac{w(o_j)}{w(i)}|\B(i) \setminus \E_{<i}| \notag \\
&=&\left(1-\frac{\sum_{j\in \mathsf{R}(i)}w(o_j)}{w(i)}\right)|\B(i) \setminus \E_{<i}|.
\end{eqnarray}
\eat{
\begin{equation}\label{equ:t3}
|\B(i) \setminus \B({\mathcal{O}})| \geq \left(1-  \sum_{j\in \mathsf{R}(i)}\frac{w(o_j)}{w(i)}\right)|\B(i) \setminus \E_{<i}|.
\end{equation}
}
By (\ref{equ:choose}) we also have
\[
|\B(i) \setminus \E_{<i}| \geq w(i) \cdot \max_{t\leq j\leq m}\frac{|\B({o_j}) \setminus \E_{<i}|}{w(o_j)}\geq w(i) \cdot\frac{\sum_{j=t}^{m}|\B({o_j}) \setminus \E_{<i}|}{\sum_{j=t}^{m}w(o_j)},
\]
where we use the inequality $\max_j\frac{A_j}{B_j}\geq \frac{\sum_j
A_j}{\sum_j B_j}$ when $B_j\geq 0$ for all $j$. Plugging this
inequality into (\ref{equ:t3}) and noting that $\E_{<i} \subseteq
\E({\mathcal{S}})$, we obtain:
\begin{eqnarray}\label{equ:t4}
|\B(i) \setminus \B({\mathcal{O}})| &\geq& \left(1-\frac{\sum_{j\in \mathsf{R}(i)}w(o_j)}{w(i)}\right)w(i) \cdot\frac{\sum_{j=t}^{m}|\B({o_j}) \setminus \E_{<i}|}{\sum_{j=t}^{m}w(o_j)}\notag \\
&=&\frac{w(i)-\sum_{j\in \mathsf{R}(i)}w(o_j)}{\sum_{j=t}^{m}w(o_j)}\cdot \sum_{j=t}^m|\B({o_j}) \setminus \E_{<i}|\notag\\
&\geq&\frac{w(i)-\sum_{j\in \mathsf{R}(i)}w(o_j)}{\sum_{j=t}^{m}w(o_j)}\cdot \sum_{j=t}^m|\B({o_j}) \setminus \E({\mathcal{S}})|.
\end{eqnarray}
Applying (\ref{equ:t4}) for all $i\in \mathcal{S}$ and summing the resulting inequalities up, we get
\begin{eqnarray}\label{equ:t5}
& &\sum_{i\in \mathcal{S}}|\B(i) \setminus \B({\mathcal{O}})|
\geq \frac{\sum_{i\in \mathcal{S}}(w(i)-\sum_{j\in \mathsf{R}(i)}w(o_j))}{\sum_{j=t}^{m}w(o_j)}\cdot \sum_{j=t}^m|\B({o_j}) \setminus \E({\mathcal{S}})| \notag\\
&=& \frac{\sum_{i\in \mathcal{S}}w(i) - \sum_{i\in \mathcal{S}}\sum_{j\in \mathsf{R}(i)}w(o_j)}{\sum_{j=t}^{m}w(o_j)}\cdot \sum_{j=t}^m|\B({o_j}) \setminus \E({\mathcal{S}})| \notag\\
&=& \frac{w(\mathcal{S}) - \sum_{j=1}^{t-1}w(o_j)}{\sum_{j=t}^{m}w(o_j)}\cdot \sum_{j=t}^m|\B({o_j}) \setminus \E({\mathcal{S}})|,
\end{eqnarray}
where the last equality holds because $\{\mathsf{R}(i)\}_{i\in
\mathcal{S}}$ is a partition of $\{1,2,\ldots,t-1\}$.

Recall that we
are dealing with the case of $w(\mathcal{S}) \geq \kb$. Since
$\mathcal{O}$ is an optimal solution meeting the weight constraint,
$w(\mathcal{O})=\sum_{j=1}^{m}w(o_j) \leq \kb \leq w(\mathcal{S})$.
Therefore, by (\ref{equ:t5}) we have
\begin{eqnarray*}
\sum_{i\in \mathcal{S}}|\B(i) \setminus \B(\mathcal{O})| \geq
\frac{\sum_{j=1}^{m}w(o_j) - \sum_{j=1}^{t-1}w(o_j)}{\sum_{j=t}^{m}w(o_j)}\cdot \sum_{j=t}^m|\B(o_j) \setminus \E(\mathcal{S})| = \sum_{j=t}^m|\B(o_j) \setminus \E(\mathcal{S})|,
\end{eqnarray*}
which immediately gives (\ref{equ:goal3}). This completes the proof
of Theorem \ref{thm:robust_additive}.
\end{proof}

At the end of this section, we prove Theorem \ref{thm:robust_apx} using Theorem~\ref{thm:robust_additive} and the partial enumeration technique.
Fix a parameter $\epsilon>0$. Given an instance $\calI$ of \rknap, call a point $v\in V$ \emph{heavy} if $w(v) \geq \epsilon\cdot\kb$. Let $\calO \subseteq \V$ be the set of centers taken by the optimal solution of $\calI$ (without violating the knapsack constraint), and $\calH$ be the set of heavy centers in $\calO$. Let $\opt$ denote the optimum objective value.
Clearly, $|\calH| \leq \kb/(\epsilon\cdot\kb)=1/\epsilon$. We guess the elements of $\calH$ by trying all possible cases (at most $|V|^{1/\epsilon}=|V|^{O(1)}$ possibilities) and using the best solution. We then construct a new instance $\calI'$ of \rknap as follows: the metric space is the same as that of $\calI$, the weight function $w'$ is defined as $w'(v)=0$ for $v \in \calH$ and $w'(v)=w(v)$ for $v\in \V \setminus \calH$, and the weight budget is $\kb' = \kb - w(\calH)$. It is easy to see that opening $\calO$ in $\calI'$ gives a feasible solution of cost $\opt$. Note that the maximum weight of any point in $\calI'$ is at most $\epsilon \cdot \kb$. Hence, by Theorem \ref{thm:robust_additive}, we can find in polynomial time a solution $\calS$ such that $cost(\calS)\leq 3\opt$ and $w'(\calS)<\kb - w(\calH) + \epsilon \cdot \kb$. We use $\calS$ as our solution to the original instance $\calI$. Then, $cost(\calS) \leq 3\opt$ and $w(\calS)\leq w'(\calS) + w(\calH)<(1+\epsilon)\kb$. The proof is complete.

\section{Concluding Remarks and Open Problems}
We gave a 3-approximation algorithm for \mcenter\, and the best known inapproximability bound is $2-\epsilon$.
For \rmcenter, we give a $7$-approximation while the current best known lower bound is $3-\epsilon$ due to the hardness of robust
$k$-center with forbidden centers \cite{charikar2001algorithms}. It would be interesting to close these gaps.
(Note that \mcenter\ includes as a special case the $k$-center problem with forbidden centers, i.e., some points are not allowed to be chosen as centers. It is known that another generalization of the latter, namely the $k$-supplier problem, is \NP-hard to approximate within $3-\epsilon$ \cite{jacm86}.)
For \rknap, it is interesting to explore whether constant factor
approximation exists while not violating the knapsack constraint.
It is also open whether there is a constant factor approximation
for the demand version (even for the unit-weight case).
Finally,
extending our results for \rknap to the multi-constraint case seems intriguing
and may require essentially different ideas.

\bibliographystyle{abbrv}
\bibliography{refLib}

\appendix

\section{Proof of Lemma~\ref{lem:group}}\label{apx:group}
Let $\calI = (\{\calS_i\}_{1\leq i\leq n}, \{w_j,\kb_j\}_{1\leq
j\leq m})$ be an instance of the \groupknap problem, for which there
exists a solution satisfying all the weight constraints; we will
call such a solution \emph{good}. Let $\calS =
\bigcup_{i=1}^{n}\calS_i$ and $w_{\max}=\max_{v\in \calS; 1\leq
j\leq m}w_j(v)$. When $m=1$, we can simply choose from each
$\calS_i$ the element $v\in \calS_i$ with the smallest $w_1(v)$. In what
follows, we assume $m\geq 2$. If there exists $1\leq j\leq m$ such
that $w_{\max}>\kb_j$, then the element having the weight $w_{\max}$
cannot appear in any good solution, and we will modify the instance
by removing it from $\calS$. Hence, we also assume that $w_{\max}
\leq \kb_j$ for all $1\leq j\leq m$.

We apply the scaling technique that has been widely used in the design of PTASs for knapsack-like problems.
Fix $\epsilon>0$, and define $A:=\epsilon \cdot w_{\max}/n$. For each $v\in \calS$, define
$$w'_j(v) = \lfloor w_j(v)/A \rfloor \textrm{~for all~}1\leq j\leq m-1, \textrm{~and~}w'_m(v)=w_m(v).$$
Also define
$$\kb'_j=\min\{\lfloor \kb_j / A \rfloor, \lfloor n^2/\epsilon \rfloor \} \textrm{~for all~} 1\leq j\leq m-1, \textrm{~and~} \kb'_{m}=\kb_m.$$
(The choice of the ``special'' index $m$ can be arbitrary; it
indicates the constraint that we wish to satisfy.) We have $w'_j(v)
\in \{0,1,\ldots, K\}$ for all $1\leq j\leq m-1$ and $v\in\calS$,
where $K = \lfloor w_{\max}/A\rfloor = \lfloor n/\epsilon \rfloor$.
Create a new instance $\calI'=(\{\calS_i\}_{1\leq i\leq n},
\{w'_j,\kb'_j\}_{1\leq j\leq m}))$. For the original instance
$\calI$, we know that there exists a good solution $\calT \subseteq
\calS$. Using the inequality $\lfloor a\rfloor + \lfloor b \rfloor
\leq \lfloor a+b\rfloor$, we obtain that for each $1\leq j\leq m-1$,
\begin{eqnarray*}
w'_j(\calT)&=&\sum_{v\in\calT}\lfloor w_j(v)/A\rfloor\\
&\leq& \min\{\lfloor \sum_{v\in\calT}w_j(v)/A \rfloor, n\cdot \lfloor n/\epsilon \rfloor\}  \\
&\leq& \min\{\lfloor \kb_j/A\rfloor, \lfloor n^2/\epsilon \rfloor\} \\
&=& \kb'_j.
\end{eqnarray*}
Also, $w'_m(\calT)=w_m(\calT)\leq \kb_m=\kb'_m$. Therefore, $\calT$
is also a good solution of $\calI'$.
For $i\in\{1,2,\ldots,n\}$, a subset $\calT \subseteq \calS$ is
called \emph{$i$-standard} if $\calT$ consists of exactly one
element from each of the $i$ sets $\calS_1, \calS_2, \ldots,
\calS_i$. Thus a solution of $\calI'$ is just an $n$-standard
subset, and vice versa. For each tuple $(i, p_1, p_2, \ldots,
p_{m-1})$ where $i\in\{1,\ldots,n\}$ and $(\forall 1\leq j\leq
m-1)p_j \in \{0,1,\ldots,\kb'_j\}$, let
$F(i,p_1,p_2,\ldots,p_{m-1})$ denote the minimum possible value of
$p_m$ for which there exists an $i$-standard subset $\calT$ such
that $w_j(\calT)\leq p_j$ for all $1\leq j\leq m$, and let $\calT(i,
p_1, p_2, \ldots, p_{m-1})$ be an (arbitrary) such $i$-standard
subset. If such $p_m$ does not exist, then we let
$F(i,p_1,p_2,\ldots,p_{m-1}) = \infty$ and $\calT(i, p_1, p_2,
\ldots, p_{m-1})=\emptyset$. Since $\calI'$ admits a good solution,
it is easy to see that
$$F(n,\kb'_1,\kb'_2,\ldots,\kb'_{m-1}) \leq \kb'_m.$$
Our goal is thus to find $\calT(n,\kb'_1,\ldots,\kb'_{m-1})$. Note
that the number of tuples $(i,p_1,\ldots,p_{m-1})$ is at most $n
\cdot \prod_{j=1}^{m-1}\kb'_j \leq n(n^2/\epsilon)^{m-1}=n^{O(1)}$,
since $m$ and $\epsilon$ are both constants.

We now compute all $F(i,p_1,p_2,\ldots,p_{m-1})$ and find the corresponding $i$-standard subsets by dynamic programming.
The base case is $i=1$. For each tuple $(1, p_1, p_2, \ldots, p_{m-1})$, let $\mathcal{R}=\{v\in \calS_1~|~(\forall 1\leq j\leq m-1)w_j(v)\leq p_j\}$.
If $\mathcal{R}\neq \emptyset$, then clearly $F(1,p_1,\ldots,p_{m-1})=\min_{v\in \mathcal{R}}w_m(v)$, and we set $\calT(1,p_1,\ldots,p_{m-1})$ to be the vertex $v\in\mathcal{R}$ that achieves the minimum $w_m(v)$. If $\mathcal{R}=\emptyset$, then $F(1,p_1,\ldots,p_{m-1})=\infty$ and $\calT(1,p_1,\ldots,p_{m-1})=\emptyset$.

Next we derive the transition function for computing $F(i,p_1,p_2,\ldots,p_{m-1})$ for $i\geq 2$. We enumerate all possible $v\in \calS_i$ that may belong to $\calT(i,p_1,\ldots,p_{m-1})$. Then, it is easy to see that
$$F(i,p_1,\ldots,p_{m-1})=\min_{v\in \calS_i}\{w_m(v)+F(i-1,p_1-w_1(v),p_2-w_2(v),\ldots,p_{m-1}-w_{m-1}(v))\}.$$
(We assume $F(i',p'_1,\ldots,p'_{m-1})=\infty$ if $p'_j<0$ for some $j$.)

If $F(i,p_1,\ldots,p_{m-1})=\infty$, then we let $\calT(i,p_1,\ldots,p_{m-1})=\emptyset$; otherwise, assuming the minimum value is attained at $v\in\mathcal{S}_i$, we set
$$\calT(i,p_1,\ldots,p_{m-1}) = \{v\} \cup \calT(i-1,p_1-w_1(v),\ldots,p_{m-1}-w_{m-1}(v)).$$

In this way, we can correctly compute the values of every $F(i,p_1,\ldots,p_{m-1})$ and find the set $\calT(i,p_1,\ldots,p_{m-1})$ witnessing the value. Since there are only $n^{O(1)}$ tuples and the time spent on each tuple is polynomial in the number of elements, the computation can be done in polynomial time.

As argued before, $\calT = \calT(n,\kb'_1,\kb'_2,\ldots,\kb'_{m-1})$
is a good solution to $\calI'$, provided that the original instance
$\calI$ has a good solution. Now we take $\calT$ as our solution to
$\calI$. (We note that, if the original instance $\calI$ is not
guaranteed to have a good solution, then we may have
$F(n,\kb'_1,\ldots,\kb'_{m-1})>\kb'_m$, in which case we will simply
return an empty set. This can happen when Algorithm
\ref{alg:multi_knap} is executed with an incorrect value of $\opt$.)
We have $w_m(\calT)=w'_m(\calT)\leq \kb'_m=\kb_m$. For each $1\leq
j\leq m-1$, $w'_j(v)=\lfloor w_j(v)/A \rfloor>w_j(v)/A-1$, and thus
we have
\begin{eqnarray*}
\sum_{v\in\calT}w_j(v) &\leq& \sum_{v\in\calT}(A \cdot w'_j(v) + A) = A \cdot \sum_{v\in\calT}w'_j(v)+nA\\
&\leq& A \cdot \kb'_j + n\cdot \epsilon\cdot w_{\max}/n\\
&\leq& \kb_j + \epsilon \cdot w_{\max}\\
&\leq& (1+\epsilon)\kb_j ~~\textrm{(since~}w_{\max} \leq \kb_j).
\end{eqnarray*}
Therefore, $\calT$ is a solution of $\calI$ that satisfies one of
the constraints and violates the others by at most a factor of
$1+\epsilon$. (It is easy to see that, by modifying the definitions
of $\{w'_j\}$ and $\{\kb'_j\}$, we can make any one of the
constraints to be the satisfied one.) The proof of Lemma
\ref{lem:group} is thus complete.

\end{document}